\documentclass{article}
\usepackage{amssymb,amsthm,hyperref,graphicx,xspace,hyperref,microtype,wrapfig,fullpage}
\hypersetup{colorlinks=false, pdfborder={0 0 0}}
\usepackage[numbers,sort&compress]{natbib}

\newtheorem{theorem}{Theorem}[section]
\newtheorem{lemma}[theorem]{Lemma}
\newtheorem{corollary}[theorem]{Corollary}
\newtheorem{claim}{Claim}
    \theoremstyle{definition}
    \newtheorem{definition}[theorem]{Definition}
    \theoremstyle{remark}
    \newtheorem{remark}{Remark}

\newcommand{\e}{\emph}
\renewcommand{\cal}[1]{\ensuremath{\mathcal{#1}}\xspace}
\newcommand{\bb}[1]{\ensuremath{\mathbb{#1}}\xspace}
\renewcommand{\-}{\textrm{-}}

\newcommand{\m}[2][1]{\ensuremath{{<{#2}>}^{({#1})}}\xspace}
\renewcommand{\P}[1][1]{\ensuremath{P^{#1}}\xspace}
\newcommand{\bd}[1]{\ensuremath{\mathrm{bd}{#1}}\xspace}

\renewcommand{\l}{\ensuremath{L}\xspace}
\newcommand{\E}{\ensuremath{\cal{E}}\xspace}

\newcommand{\G}{\ensuremath{\cal{G}}\xspace}
\newcommand{\mc}{\ensuremath{m_{\cal{C}}}\xspace}

\newcommand\wout{\ensuremath{\cal{W}_{out}}\xspace}
\newcommand\win{\ensuremath{\cal{W}_{in}}\xspace}
\newcommand\w{\ensuremath{\cal{W}}\xspace}

\title{\bfseries Simple Wriggling is Hard\\ unless You Are a Fat Hippo\thanks{A shorter version of this paper is to be presented at the Fifth International Conference on Fun with Algorithms~\cite{fun}.}
}

\author{Irina Kostitsyna\thanks{CS Dept, Stony Brook University, Stony Brook, NY 11794, USA. 
} \and Valentin Polishchuk\thanks{Helsinki Institute for Information Technology,
 CS Dept, University of Helsinki, FI-00014, Finland. 
 }}
\date{}

\begin{document}\maketitle\begin{abstract}
We prove that it is NP-hard to decide whether two points in a polygonal domain with holes can be connected by a wire. This implies that finding any approximation to the shortest path for a long snake amidst polygonal obstacles is NP-hard. On the positive side, we show that snake's problem is "length-tractable": if the snake is "fat", i.e., its length/width ratio is small, the shortest path can be computed in polynomial time.
\end{abstract}
\section{Introduction}
\begin{wrapfigure}{r}{.5\columnwidth}\centering
\includegraphics[width=.5\columnwidth]{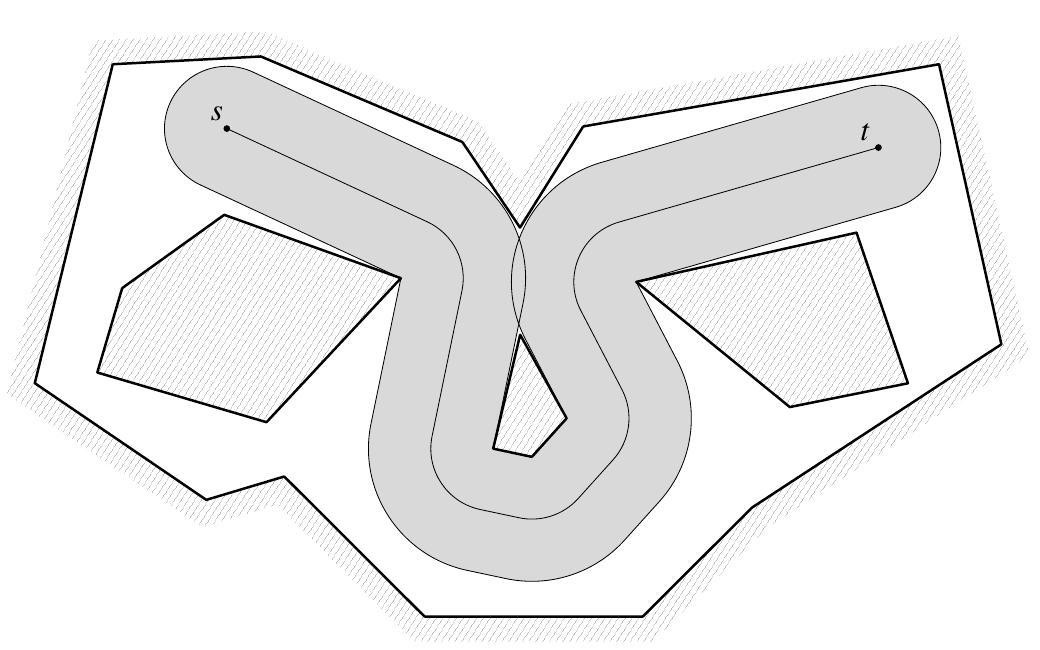}
\caption{Only a self-overlapping thick path exists.}\label{self}
\end{wrapfigure}

The most basic problem in VLSI and printed circuit board design is to connect two given points, $s$ and $t$, by a shortest "thick" path avoiding a set of polygonal obstacles in the plane. The quarter-of-a-century-old approach to the problem is to inflate the obstacles by half the path width, and search for the shortest $s\-t$ path amidst the inflated obstacles \cite{chew}. The found path, when inflated, is the shortest thick $s\-t$ path.

It went almost unnoticed that the thick path built by the above procedure may self-overlap (Fig.~\ref{self}): apart from our recent work on thick paths \cite{arkin10maximum}, we only found one mention of the possibility of the overlap --- Fig.~4 in \cite{hu}. (In a different context, Bereg and Kirkpatrick \cite[Fig.~2]{corridors} also noted that Minkowski sum of a disk and a path may be not simply-connected.) When the path represents a thick \e{wire} connecting terminals on a VLSI chip or on a circuit board, self-overlap is undesirable as the wire must retain its width throughout. Thus, the objective in the basic wire routing problem should be to find the shortest \e{non-selfoverlapping} thick path.

The problem shows up in other places as well. For instance, one may be interested in the optimal conveyor belt design: the belt is a non-selfoverlapping thick path. Our particular motivation comes from air traffic management where thick paths represent lanes for air traffic. Lane thickness equals to the minimum lateral separation standard, so that aircraft following different lanes stay sufficiently far apart to allow for errors in positioning and navigation. If an airlane self-overlaps, two aircraft following the lane may come too close to each other; thus it is desirable to find lanes without self-overlaps.
\subsection{Our contributions}
We prove a surprisingly strong {\textbf{\textit{negative result}}} (Section~\ref{hard}): it is NP-hard even to \e{decide} whether there exists (possibly, arbitrarily long) $s\-t$ wire; this implies that no approximation to the shortest wire can be found in polynomial time (unless P=NP).
\paragraph{Short Snakes} Our intractability result means that in general it is NP-hard for a snake to wriggle its way amidst polygonal obstacles (assuming the snake is uncomfortable with squeezing itself). The good news for snakes is that in our hardness proof the sought wire is considerably long; i.e., the hardness of path finding applies only to \e{long} snakes. Our {\textbf{\textit{positive result}}} (Section~\ref{short}) is that for a bounded-length snake, the shortest path can be found in polynomial time (assuming real RAM and the ability to solve constant-size differential equations in constant time) by a Dijkstra-like traversal of the domain.
\subsection{Related work}
In VLSI numerous extensions and generalizations of the basic problem were considered. These include routing multiple paths, on several levels, and with different constraints and objectives. It is impossible to survey all literature on the subject; we will only mention the books \cite{vlsiBook,maley}.

In robotics thick paths were studied as routes for a circular robot. In this context, path self-overlap poses no problem as even a self-overlapping path may be traversed by the robot; that is, in contrast to VLSI, robotics research should not care about finding non-selfoverlapping paths. In \cite{chew}, Chew gave an efficient algorithm for finding a shortest thick path in a polygonal domain. In a sense, our algorithm for shortest path for a short snake (Section~\ref{short}) may be viewed as an extension of Chew's.

Motion planning for an object with few degrees of freedom may be approached with the \e{cell decomposition} techniques \cite{latombe,lavalle}. Closest to our bounded-length snake problem is the work on path planning for a \e{segment} (rod) \cite{rod,rod3d,rodCCCG,rodSharir}. Short snakes are also relevant to more recent applications of motor protein motion \cite{protein,nature}.

\section{Snake Anatomy and Physiology}
In this section we introduce the notation and formulate our problem.

Let $P$ be an $n$-vertex polygonal domain with obstacles. For a planar set $\cal{S}$ let \bd{\cal{S}} denote the boundary of $\cal{S}$, and for $r>0$ let \m[r]{\cal{S}} denote the Minkowski sum of \cal{S} with the radius-$r$ open disk centered at the origin. Let $\P[r]=P\setminus\m[r]{\bd{P}}$ be $P$ offset by $r$ inside. The boundary of \P[r] consists of straight-line segments and arcs of circles of radius $r$ centered on vertices of $P$. We call such (maximal) arcs \e{$r$-slides}.

Let $\pi$ be a path within \P; let $|\pi|$ denote its length. A {\em thick path} $\Pi$ is the Minkowski sum $\Pi=\m{\pi}$. The path $\pi$ is called the \e{reference path} of $\Pi$; the {\em length} of $\Pi$ is $|\pi|$.

A \e{snake} is a non-selfoverlapping thick path, i.e., a path which is a simply-connected region of the plane. The reference path of the snake is its \e{spine} (Fig.~\ref{anatomy}). One of the endpoints of the spine is the snake's \e{mouth} $m$. The snake is a ``rope'' that ``pulls itself by the head'': imagine that there are little legs (or a wheel, for a toy snake) located at $m$, by means of which the snake moves. The friction between the snake's body and the ground is high: any point $p$ of the spine will move only when the path from the mouth to $p$ is a ''pulled-taut string``, i.e., is a \e{locally shortest} path. That is, the snake always stays pulled taut against the obstacles (or itself).
\begin{figure}\centering
\includegraphics[width=.8\columnwidth]{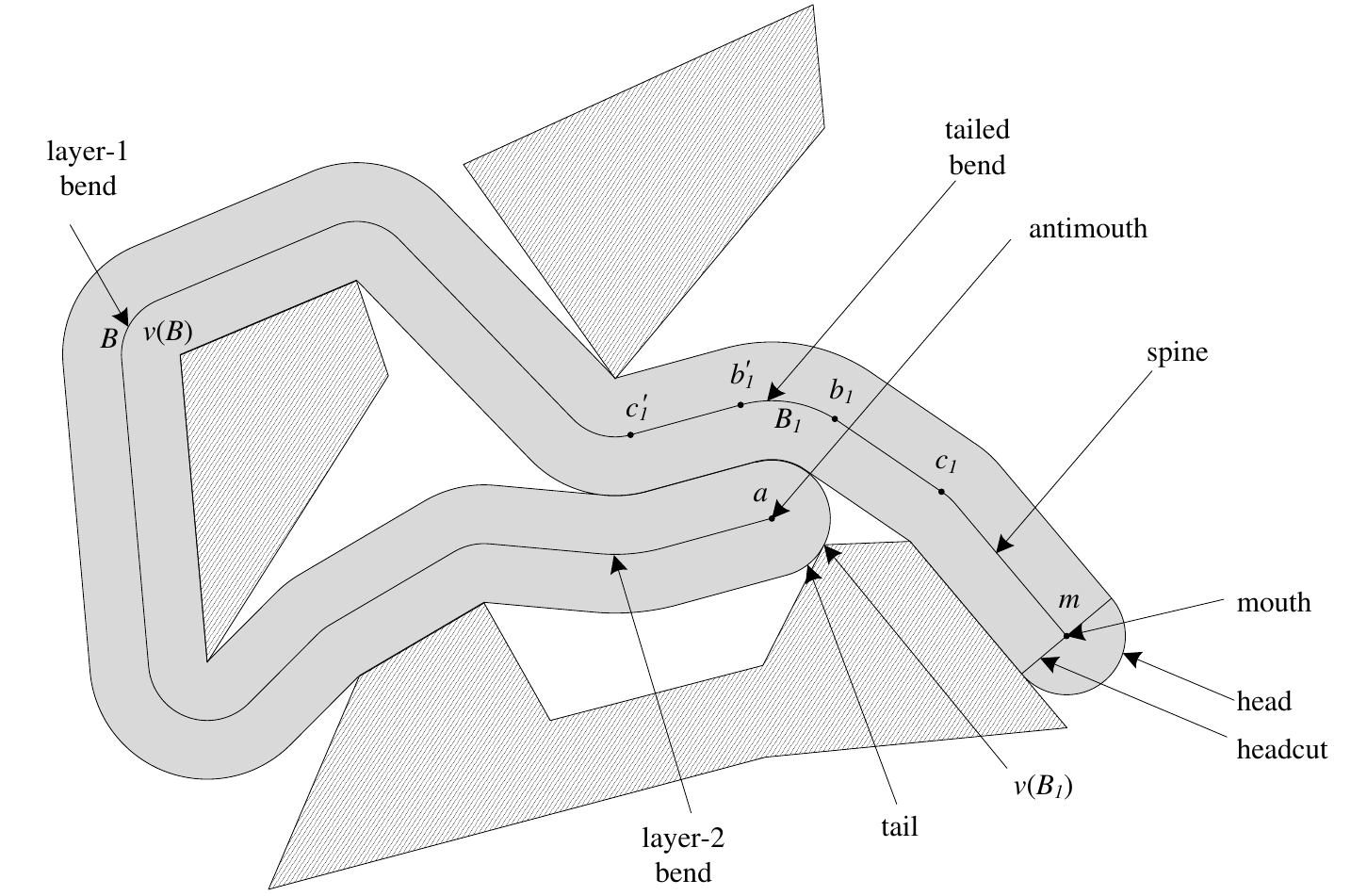}
\caption{Spine, mouth, head, headcut, antimouth, tail, bends. A layer-2 tailed bend $B_1$ does not have a point connected to $v(B_1)$ with a length-2 segment fully lying within the snake.}\label{anatomy}
\end{figure}

The input to our problem is the domain $P$, two points $s,f\in\P$ -- the ``start'' and the ``food'', a number $\l>0$ -- the length of the snake, and the initial direction of the snake at $s$. (Assume w.l.o.g.\ that $s$ and $f$ are at distance 1 from some vertex of $P$.). The goal is to find a path for the snake such that the snake's mouth starts at $s$ and ends at $f$; the constraints are that the snake remains pulled taut and non-selfoverlapping throughout the motion. The objective is to minimize the distance traveled by the mouth, or equivalently, assuming constant speed of motion, the time until the snake reaches the food.
\begin{remark}
To simplify the formulation, we did not specify the initial configuration of the snake; a pedantic view could be to assume that the snake slithers in from a Riemann sheet glued along the diameter of \m{s} that is perpendicular to the initial direction of the snake's motion. Our results, both positive and negative, remain valid even if the initial configuration of the snake is part of the input.
\end{remark}
\begin{remark}
It is not true that all points of the spine necessarily follow the same path.
\end{remark}
\begin{remark}
Whoever guesses that the above model of a snake was developed just for FUN, is right. Nevertheless, the proposed problem formulation may be relevant also in more serious circumstances. It models, e.g., the path of a rope being pulled by its frontpoint through a polygonal domain. If it is a fire hose or a tube delivering life-saving medicine \cite{needleWAFR,needleCook}, minimizing the time to reach a certain point seems like a natural objective (more important than, say, the work spent on pulling the tube).
For another application, consider a chain of robots moving amidst obstacles. Each robot, except for the leader of the chain, has very simple program of following its predecessor -- just keeping the distance to it. Then the robots form (approximately) a pulled taut thick string.
\end{remark}
\begin{remark}
It is possible to come up with problem instances in which no path for a pulled-taut snake exists, while a path for a snake that is not required to be taut, does.
\end{remark}

\section{Shortest Path for a Fat Hippo}\label{short}
When the snake is relaxed and its spine is a straight-line segment, the snake is the Minkowski sum of the length-\l segment and the unit disk. Such sums are known as \e{hippodromes} \cite{hippoPach,hippoAlon,hippoanchored,hippoPankaj,hippoAlon2}, or \e{hippos} for short. We say that a hippo is \e{fat} if its length is constant: $\l=O(1)$. In this section we show that for a fat hippo our problem can be solved in polynomial time.
\subsection*{Overview of the approach}
Our algorithm is a Dijkstra-style search in an implicitly defined graph \G (Fig.~\ref{overview}): neither the nodes nor the edges of the graph are known in advance. Instead, \G is built incrementally, by propagating the labels from the node $v$ with the smallest temporary label (as in standard Dijkstra). The labels are propagated to nodes in the ''\l-visibility`` region of $v$, which is what the snake ''sees`` while it slithers for distance \l starting from $v$. In order to discover the nodes in the region, we pull the snake from $v$ for length \l along every possible combinatorial type of path; by a packing argument, there is only a small number of the types. The edges of \G correspond to bitangents between the paths and 1-slides. We prove that the algorithm is polynomial-time by observing that the snake must travel for at least $2\pi-2$ before it ``touches'' itself with its head; this implies that the snake never ``covers'' any point of \bd{P} with more than $O(\l/\pi)=O(1)$ ``layers'', and hence there is a polynomial number of relevant bitangents.
\begin{figure}\centering
\includegraphics{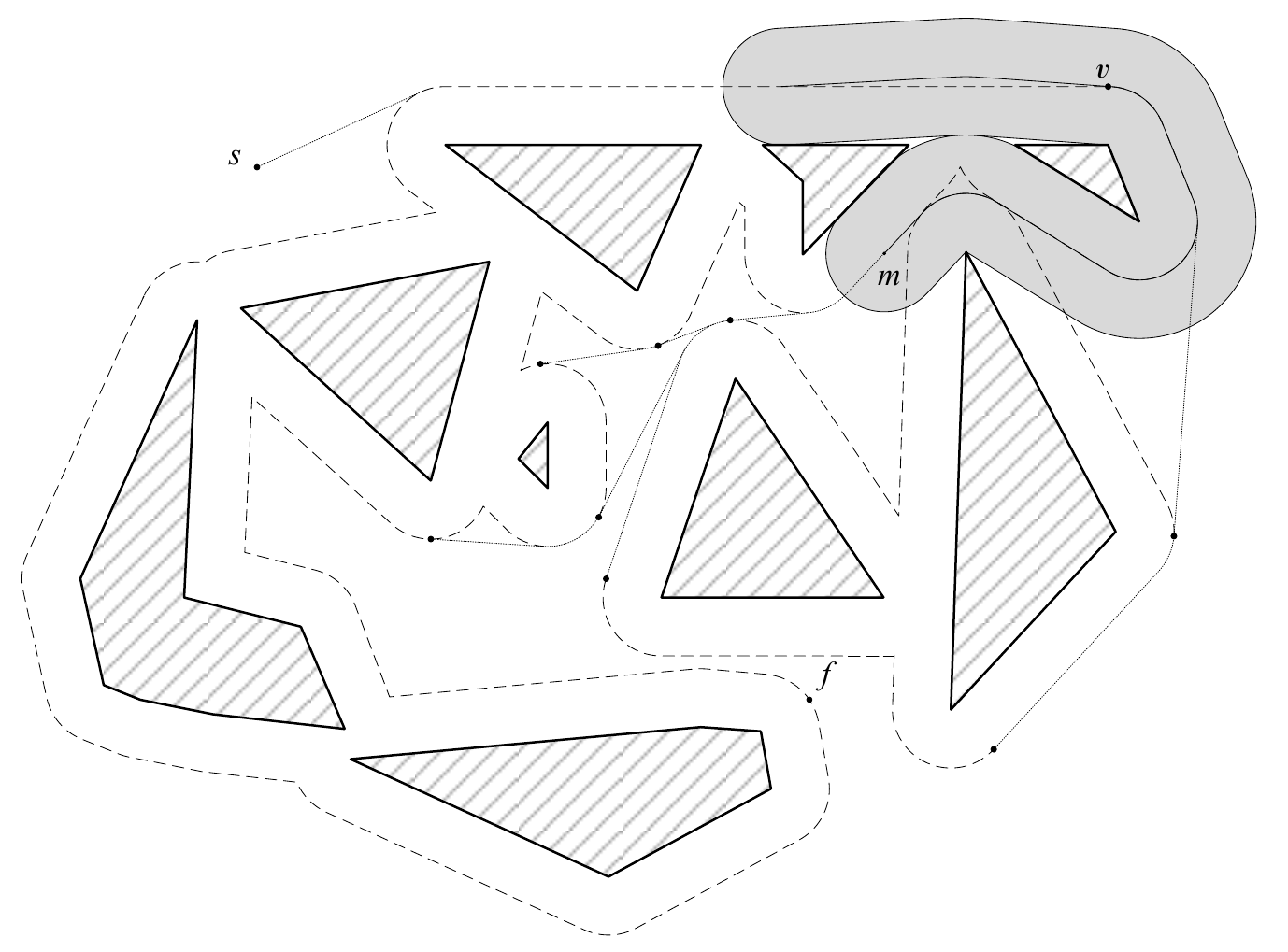}
\caption{$v$ is a node of \G; pull the snake from $v$ for distance \l along "every possible" path. Endpoints of visibility edges (some of them are shown with solid circles) become nodes of \G.}\label{overview}
\end{figure}

In the reminder of the section we elaborate on the algorithm's details.
\subsection{A Bit More Anatomy}
Consider the unit disk \m{m}. The part of the boundary of \m{m} that is also the boundary of the snake is the unit semicircle whose diameter is perpendicular to the spine at $m$; we call the part the \e{head} and the diameter the \e{headcut}. The endpoint of the spine that is not the mouth is called the \e{antimouth} $a$. The part of the boundary of \m{a} that is also the boundary of the snake is called the \e{tail}.
Refer to Fig.~\ref{anatomy}.
\subsection{Freeze!}\label{freeze}
Let us have a closer look at the structure of the pulled-taut snake at any moment of time. Some pieces of the spine are straight-line segments. The segments are bitangents between the other, non-segment pieces supported by vertices of $P$, possibly via several ``layers'' of the snake. We call each such (maximal) piece $B$ a \e{bend}; we denote the vertex that supports $B$ by $v(B)$, and say that $v$ is \e{responsible} for~$B$.
\paragraph{Snake layers}
We show that each vertex $v$ is responsible for $O(1)$ bends. Say that a bend $B$ belongs to \e{layer} 1 if there exists a point $b\in B$ such that $|bv(B)|=1$. Recursively, $B$ is \e{layer-$(k+1)$} bend if it is not layer-$k$ and there exists a point $b\in B$ such that for some point $b'$ at layer $k$, we have $|b'b|=2$ and $b'b$ fully lies within the (closure of the) snake (Fig.~\ref{anatomy}). Let $K$ be the maximum index of a layer.
\begin{lemma}\label{dubins}
$K\le\lceil\l/(2\pi-2)\rceil$.
\end{lemma}
\begin{proof}
Let $b_k,b_{k+1}$ be points on layers $k,k+1$ such that $|b_kb_{k+1}|=2$. Let $\pi(b_kb_{k+1})$ be the part of the spine between $b_k$ and $b_{k+1}$. The Minkowski sum \m{\pi(b_kb_{k+1})} encloses at least one obstacle, $h$ (or else the snake is not pulled taut). The perimeter of the inflated obstacle \m{h} is at least $2\pi$, thus $|\pi(b_kb_{k+1})|+|b_{k+1}b_k|\ge2\pi$. But $|b_{k+1}b_k|=2$, hence $|\pi(b_kb_{k+1})|\ge2\pi-2$.
\end{proof}
As a corollary from the lemma, we have that life is very simple for a fat enough hippo:
\begin{corollary}
For $\alpha\ge2\pi-2$, an $\alpha$-fat hippo can follow a shortest thick path without self-overlap.
\end{corollary}
\paragraph{Tailed and headed bends}
One may wonder why we did not opt for a simpler definition of layer-$k$ bend $B$ as one having a point $b$ such that $|bv(B)|=2k-1$ and $bv(B)$ lies fully within the (closure of the) snake. The reason are two special kinds of bends which make the snake's portrait more complicated than in the case of an infinite-length snake. Specifically, we say that a bend $B$ at layer $k$ is \e{tailed} (resp.\ \e{headed}) if \m[2k-3]{B} touches the tail (resp.\ head). The simpler definition may not work for such bends (Fig.~\ref{anatomy}).

Any bend that is not tailed or headed is an arc of a circle (actually, it is part of a slide). A tailed or headed bend $B$ is a different shape -- string pulled taut against a ball touching $v(B)$.
\paragraph{Snake configuration}
The snake can be reconstructed in linear time as soon as the following is specified: (1)~list of the vertices responsible for the bends; (2)~for each bend, its layer; (3)~for a headed bend -- the vertex or the edge of $P$ in contact with the head, and the slope of the headcut; (4)~similar information for a tailed
bend; (5)~positions of the mouth $m$ and antimouth~$a$. We will call (1)--(5) the \e{configuration} of the snake.
\begin{lemma}\label{complexity}
The list (1) contains $O(n)$ vertices.
\end{lemma}
\begin{proof}
Each vertex may be responsible for up to $K$ bends, hence the total number of bends is $O(Kn)=O(n)$. (In fact, since the part of the snake between consecutive bends leaning on one vertex, encloses at least one obstacle, every bend can be charged either to a vertex or to an obstacle; thus the number of bends is actually $O(n+h)=O(n)$, where $h$ is the number of obstacles.)
\end{proof}
\subsection{Move!}\label{move}
Suppose that we are given the configuration \cal{C} of the snake at some time; let $\pi$ be the spine when the snake is in \cal{C} and let \mc be the position of the mouth in \cal{C}. Suppose we are also given a path $\gamma$ for the mouth starting at \mc. Assume $\gamma$ has the following properties: (1) it is consistent with \cal{C} in that the tangent to $\gamma$ at \mc coincides with the tangent to $\pi$ at \mc; (2)~it has a polynomial number of pieces, each of constant description complexity; (3) $|\gamma|\le\l$. In what follows we assume that every path $\gamma$ has these properties. We claim that in polynomial time we can check whether $\gamma$ is a feasible path for the mouth, i.e., whether the snake stays obstacle-free when $m$ is pulled along $\gamma$ (say, at unit speed).

First we note that it is enough to check only whether the mouth moves feasibly. Indeed, the first time that the snake (possibly) becomes infeasible as $m$ follows $\gamma$, the mouth is necessarily a part of ``certificate of infeasibility''. This is so because the only way that the snake experiences ``side pressure'' is due to appearance of a headed bend. Other than that, any piece \cal{P} of the snake is merely pulled by the preceding piece, \cal{P'}: either \cal{P} exactly follows the same path that just was feasible for \cal{P'}, or \cal{P} is a tailed bend $B$. Of course, the bend morphs as the time passes, but only becomes ``more feasible'' with time, i.e., the free space around $B$ increases -- $B$ is pushed only by the tail, and the tail ``moves away'' with time.

Let now $\gamma'$ be a piece of $\gamma$. If we know how each piece \cal{P} of the snake changes with time, we can test whether $m$ stays feasible while following $\gamma'$, by checking the feasibility against each \cal{P} in turn. That is, while neither the configuration of the snake nor the piece $\gamma'$ of $\gamma$ changes, the feasibility test can be done piece-versus-piece in constant time (assuming real RAM). Observe that overall there is only a polynomial number of configuration changes. Indeed, the configuration may change only due to one of the following \e{events}: (1)~the tail starts to follow another feature of $P$, (2)~a headed
bend appears or changes its combinatorial structure, (3)~a tailed
bend disappears or changes its combinatorial structure. But each of the events (1)--(3) may happen only once per vertex-bend-layer triple; thus, by Lemmas~\ref{dubins} and~\ref{complexity} there is only a polynomial number of events.

Tracking the configuration changes is easy \e{given} the way each piece changes with time. For event (1), we only have to know how the tail speed changes with the time between consecutive events (the tail speed is not necessarily constant, we elaborate on it in the next paragraphs). For event (2) we check what is the first time that the head collides with a piece or when a headed
bend hits a vertex; all this can be done in polynomial time as there is only a linear number of candidate collisions. Event (3) is similar. The next event time is then the minimum of the event times over all the pieces.

It remains to show how to determine the way each piece changes. Here the crucial role is played by the headed and tailed bends; again, it is the finiteness of the snake length that makes things involved. Let $B_1=b_1b_1'$ be the first such bend counting from $m$ (Fig.~\ref{anatomy}). Assume that $B_1$ is a layer-$k$ tailed bend;
the situation with a headed bend is actually simpler (because the mouth speed is constant). Let $c_1b_1$ and $b_1'c_1'$ be the pieces adjacent to $B_1$ -- both are bitangents (straight-line segments, possibly of 0 length) to $B_1$ and adjacent bends. Every point of the spine between $m$ and $c_1$ moves at speed 1, and none of the bends before $B_1$ changes with time.

To figure out what happens after $c_1$, we have to solve a constant-size differential equation that describes the ''propagation'' of speed of motion of the spine. Specifically, let $u(\tau)$ denote the speed at which the antimouth moves at time $\tau$. Knowing $u(\tau)$ and knowing the initial position of the antimouth, we can write the antimouth position as a function of time, and hence we know \m[2k-2]{a} as a function of time. The points $b_1$ and $b_1'$ are points of tangency to \m[2k-2]{a} from $c_1$ and $c_1'$; thus knowing \m[2k-2]{a} we know how the length $|c_1b_1|+|\pi(b_1b_1')|+|b_1'c_1'|$ changes with time. Knowing that, and recalling that at $c_1$ the spine moves with unit speed, we can write what the spine speed at $c_1'$ is as a function of $\tau$.

Now, the spine speed does not change between $c_1'$ and the next point, $c_2$, that is the start of the tangent to the next headed or tailed bend, $B_2$. We perform at $B_2$ the same operations as above, and get the speed of motion of the spine past the bend $B_2$. Continuing in this fashion, in the end, after going through all bends, we obtain some expression, $\E(u(\tau))$ for the spine speed after the last bend.

Finally, to close the loop, we solve the equation
\begin{equation}\label{diff}
u(\tau)= \E(u(\tau))
\end{equation}
Since there is only a constant number of layers (Lemma~\ref{dubins}), there is only a constant number of the headed and tailed bends, and hence the expression $\E$ is a sum of a constant number of terms (each containing $u(\tau)$ and $\int\!u(\tau)d\tau$), each of constant description complexity. In our computation model, we can solve the equation for $u(\tau)$ in constant time. Substituting $u(\tau)$ back into the formulae for the different bends, we obtain the spine as a function of time, as desired.
\subsection{See!}
Assume that at some point the snake is in configuration \cal{C}. What happens next, as the snake follows the optimal path to $f$? Local optimality conditions imply that it will ``wiggle around the obstacles'' for some time and then ``shoot'' towards a vertex of $P$. We formalize this below.

\paragraph{Final piece of anatomy} The \e{eye} of the snake is collocated with the mouth. The snake can \e{see} a point $x\in P^1$ if the segment $mx$ lies fully within $P^1$ and is tangent to the spine at~$m$. Recall that $P^1=P\setminus\m{\bd{P}}$ where \m{\bd{P}} is the Minkowski sum of \bd{P} with \e{open} unit disk; hence $mx$ can go along the boundary of $P^1$ (with $x$ being, e.g., an endpoint of a slide; see Fig.~\ref{overview}). The snake itself is transparent for its eye: we do not forbid $mx$ to intersect the snake.
\begin{definition}\label{def}
A point $p$ on \bd{P^1} is \e{\l-visible} from \mc if there exists a path $\gamma$ for the mouth ending in a point $m^*$ such that $m^*$ sees $p$ and $m^*p$ is tangent to \bd{P^1} at $p$. We say that $\mc p$ is an \e{\l-visibility edge}, or an \e{\l-edge} for short. We say that $\gamma$ is the \e{wiggle segment} of $\mc p$, and that $m^*p$ is the \e{visibility segment} of the edge. (We remind that we assume $\gamma$ enjoys the properties listed in the beginning of Section~\ref{freeze}: tangent at \mc consistent with \cal{C}, polynomial-size description, length $\le\l$.)
\end{definition}
To be on a (locally) optimal path, the mouth would like to follow an \l-edge also past $m^*$. This may not be feasible due to a conflict with the snake itself. However, such conflicts can be discovered ``on-the-fly'', as the mouth attempts to move along $m^*p$. Specifically, try to move the mouth along $m^*p$, as described in Section~\ref{move}. If during the motion, the head collides with the snake, note what kind of bend the head collides with. If this bend $B$ is not tailed, adjust the path for the mouth so that it is tangent to \m[2]{B}, and follow the adjusted path. If by the time the mouth reaches \m[2]{B}, the bend $B$ is ``gone'', i.e., the head ``misses'' the snake, we know that we are dealing with a tailed bend (or with the tail itself). We identify the time and place of contact with the bend by solving a differential equation similar to (\ref{diff}): assuming the speed of the tail is $u(\tau)$, we know how the bitangent between $m$ and the corresponding ball centered at $a$ changes; in particular, we know its length $l(\tau)$ as a function of time. The time $\tau^*$ when the head hits the tail is then the solution to the equation $\tau^*=l(\tau^*)$. After solving the equation we know the configuration of the snake at $\tau^*$, and continue moving the mouth to $p$ around the tail.


The above procedure essentially ``develops'' the path $\gamma$ piece-by-piece. This is consistent with Section~\ref{move} where we described how to pull the snake along a \e{given} path $\gamma$ for the mouth: the pulling was done piece-by-piece, which means that it can be performed even if $\gamma$ is not given in advance but instead is revealed piece after piece. More importantly, using the procedure, we can develop \e{all} \l-edges incident to \mc, piece-by-piece, in a BFS manner. We describe this below.

To initiate the developing, look at the visibility segment $m'\mc$ of the \l-edge that has led the mouth to \mc. The segment is tangent to a slide $S\ni\mc$. The slide $S$ down from \mc (i.e., in the direction consistent with $m'\mc$) is the first (potential) piece of a new \l-edge. We extend bitangents from the piece to all other 1-slides and to all pieces of the spine inflated by 2. These bitangents become the next potential pieces for the \l-edges. After the bitangents, the next potential pieces are the slides and the spine pieces at which the bitangents end. We continue in this way (possibly adjusting the pieces of a particular edge to account for snake self-interaction) until for each edge its wiggle segment reaches length \l.

We now bound the time spent on developing all \l-edges from \mc. Each edge has linear complexity; this can be proved identically to Lemma~\ref{complexity}.
Thus all edges can be grown in polynomial time if the number of edges is polynomial. This is what we prove next:
\begin{lemma}\label{ledges}
Let \bb{E} be the set of all \l-edges incident to \mc. $|\bb{E}|=O(n^2)$.
\end{lemma}
\begin{proof}
By definition, every edge starts from a path $\gamma$ for the mouth of length at most \l. Say that paths $\gamma_1,\gamma_2$ are the same \e{combinatorial type} if the sequence of vertices visited by $\gamma_1$ is a subsequence of that for $\gamma_2$ (or vice versa). Being of the same type is an equivalence relation that splits \bb{E} into classes. For each class, keep only the path with the longest sequence of visited vertices, and identify the class with the path. Let \bb{E^*} be the obtained collection of classes.

Let $\gamma^*\in\bb{E^*}$. Any \l-edge that has a path $\gamma\in\gamma^*$ as its wiggle segment is obtained by extending a bitangent from $\gamma$ (or equivalently from $\gamma^*$) to some slide. Thus the total number of \l-edges having a path in $\gamma^*$ as the wiggle segment is at most the number of bitangents from $\gamma^*$ to the slides, which is $O(n^2)$ since $\gamma^*$ is $O(n)$-complexity.

It remains to show that $|\bb{E^*}|=O(1)$. A standard argument shows that $|\bb{E^*}|$ depends only on the number of the \e{holes} in $P^1$ reachable by length-\l paths from \mc (not on the number of vertices): Restrict attention to radius-\l disk \cal{D} centered on \mc, do vertical trapezoidation of $P^1\cap\cal{D}$, and let $G$ be the dual graph of the trapezoidation. The number of nodes of $G$ of degree higher than 2 is linear in the number of holes (in $P^1$) that intersect \cal{D}; call this number $H$. Transform $G$ to a graph $G'$ by replacing each path in $G$ by an edge between degree-3 and higher nodes; the number of nodes in $G'$ is $O(H)$. Each path from \bb{E^*} is a walk in $G'$ that visits every edge of $G'$ at most $K=O(1)$ times. The number of such walks $G'$ depends only $H$, and we prove now that $H$ is constant:
\begin{claim}\label{H}
Let $h_1\dots h_m$ be the holes (of $P$) that intersect \cal{D}. The Minkowski sum $\m{h_1\cup\dots\cup h_m}$ has a constant number of connected components.
\end{claim}
\begin{proof}
If $h_1,h_2$ are in different connected components of the sum, the distance between the holes is at least 2. Within constant distance of \mc, there can be only a constant number of points pairwise-separated by distance at least 2; thus there is only a constant-size set of holes pairwise in different connected components.
\end{proof}
This competes the proof of the lemma.
\end{proof}
\subsection{Label!}
We are ready now to traverse the "\l-visibility" graph \G of $P$, searching only the relevant part of \G and not building the whole graph explicitly. The label of each node in \G consists of two parts: the \e{distance label} (storing the distance from $s$) and the \e{configuration label} (storing the snake configuration at which the node was reached). That is, a node may have several labels -- one per configuration. However, since there is only a constant number of different configurations of pulled-taut snakes that may reach the node (Lemma~\ref{ledges} and Claim~\ref{H}), the total number of labels of any node is constant.

Start from $\mc=s$, and assign distance label 0 to $s$. The algorithm grows the graph \G whose edges are the discovered \l-edges and whose nodes are the endpoints of the \l-edges. Note that by the definition of \l-visibility (Definition~\ref{def}), all nodes of \G reside on slides. At a generic step, take the node with the smallest (temporary) distance label, make the label permanent, and construct \l-edges from the node. The endpoints of the edges join \G and get their (temporary) distance labels and configuration labels. In addition, each already existing node over which the mouth passes, gets its distance label updated if the distance label carried with the mouth is smaller than the node's current label \e{and} the configuration label carried with the mouth is the same as the node's configuration label. The search stops when $f$ is reached.

We now prove that the algorithm terminates in a polynomial number of steps. The visibility segments are bitangents between 1-slides and pieces of the wiggle segments. These latter pieces are of two types: (1)~slides up to layer $K$ and (2)~curves that are obtained as the head rolls over the tail, possibly padded by up to $K$ layers of the snake. There is in principle an uncountable number of possible pieces of the second type. Nevertheless, every step of the algorithm discovers at least one new visibility segment tangent to a piece of the first type. Since the number of slides up to layer $K$ is $O(n)$, there are $O(n^2)$ of such visibility segments. Each such segment may be discovered only a constant number of times --- once per homotopy type of the snake reaching the endpoint of the segment. Thus overall there are $O(n^2)$ steps.

Overall, we have our main positive result:
\begin{theorem}
Shortest path for a fat hippo can be computed in polynomial time.
\end{theorem}

\section{Being a Long Snake is Hard}\label{hard}
In this section we prove that if the snake length is not bounded, deciding existence of a path for the snake is NP-hard. Specifically, our problem is: Given polygonal domain $P$ and points $s,f$, find a thick non-selfoverlapping $s\-f$ path (i.e., a thick $s\-f$ wire).

We show the problem's hardness by a reduction from planar 3SAT. Recall that the \e{graph} of a 3SAT instance has nodes for all variables and clauses, and two types of edges: (1)~a cycle $\mathbb{C}$ through all variables, and (2)~edges connecting every clause to its three variables (Fig.~\ref{1}). \e{Planar 3SAT} is a restriction of 3SAT to instances whose graph is planar; Lichtenstein \cite{p3sat} proved that planar 3SAT is NP-hard. To show the hardness of wire routing, we start from an instance $I$ of planar 3SAT; we identify the instance with its (planar) graph, and the variables and clauses with the points in the plane into which they are embedded.
\begin{figure}\centering
\includegraphics{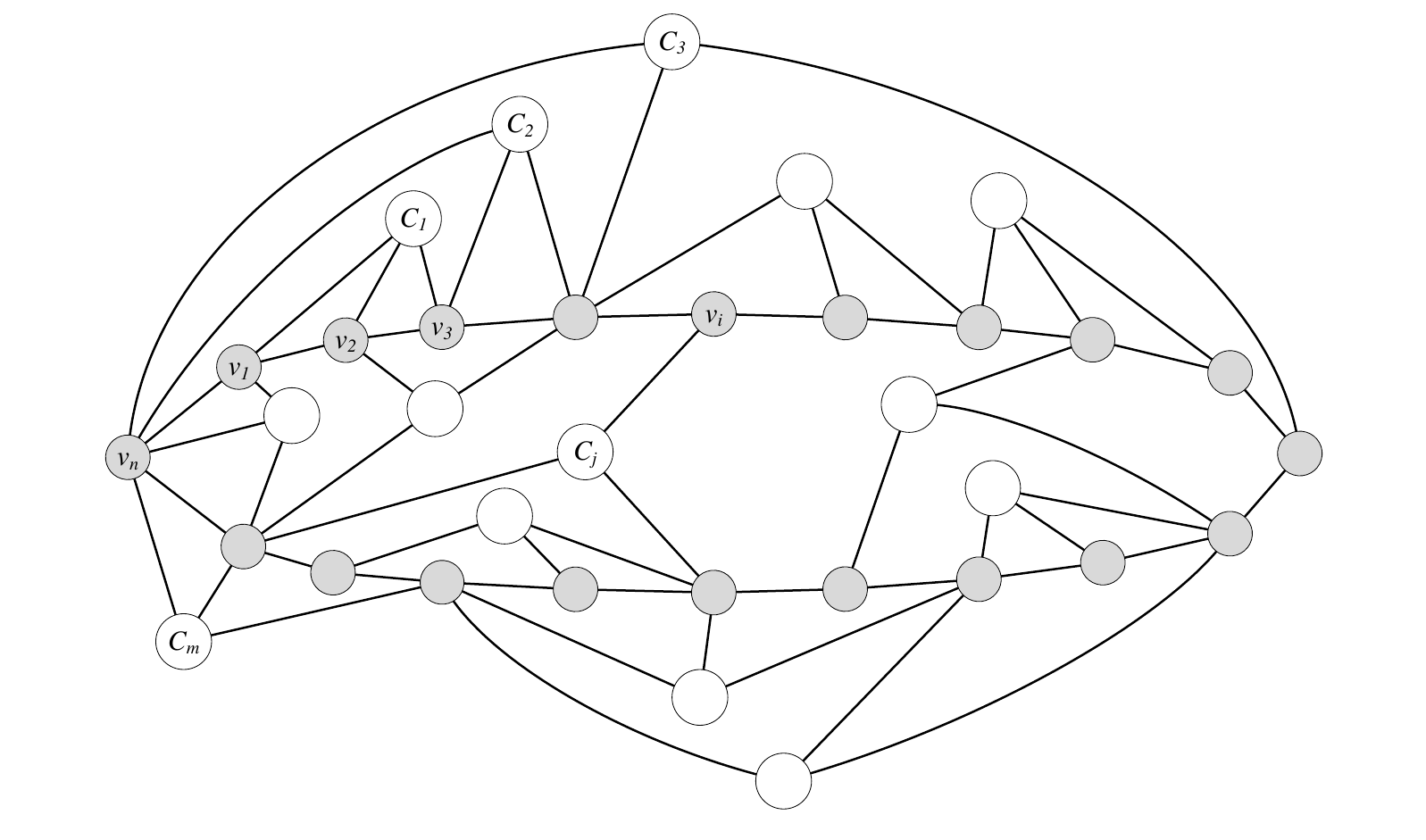}
\caption{(The graph of) an instance $I$ of 3SAT. The variables are shaded circles, the clauses are hollow circles.}\label{1}
\end{figure}
\begin{figure}\centering
\includegraphics{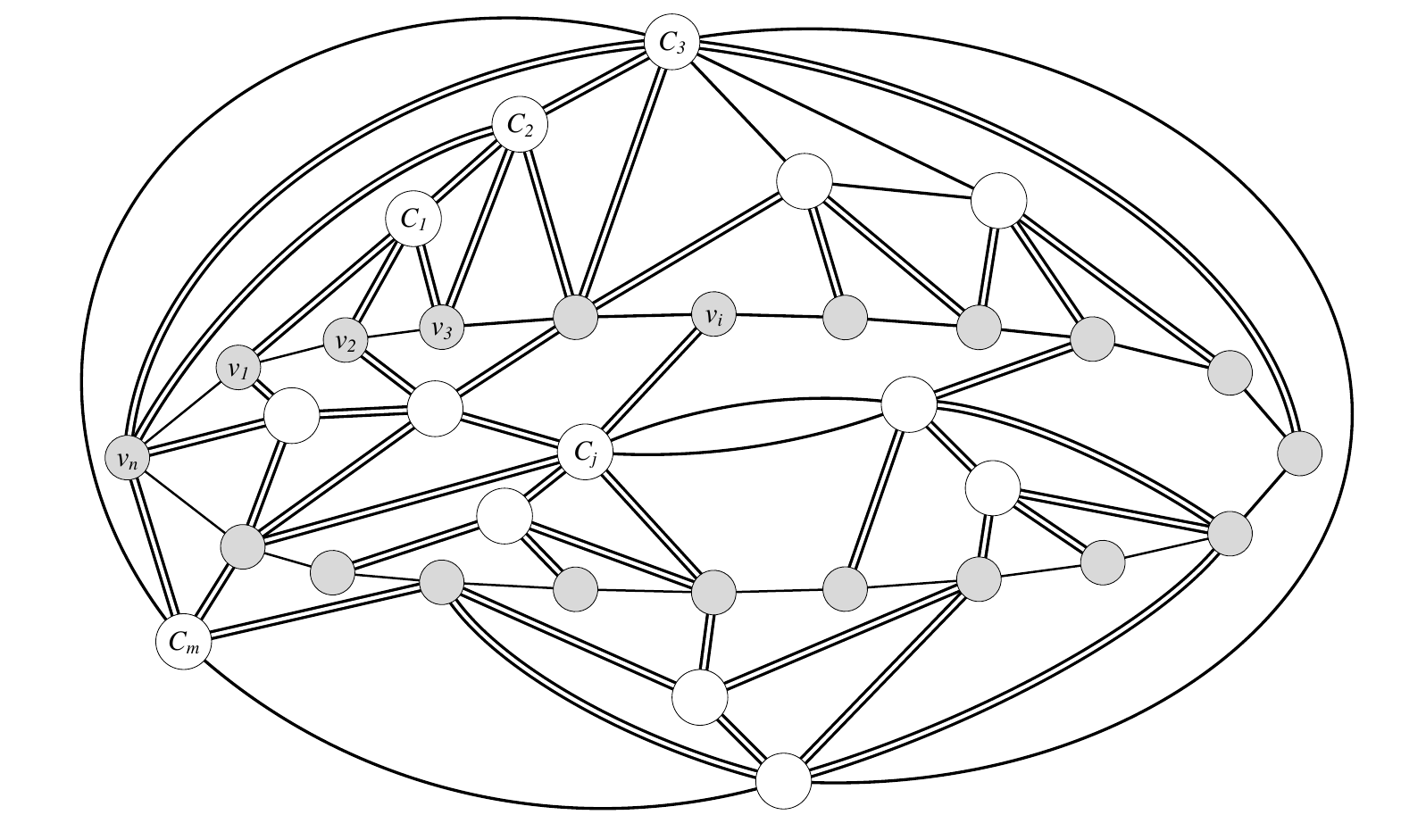}
\caption{$I$ augmented with parent-child edges, sibling edges, and with variable-clause edges duplicated.}\label{2}
\end{figure}
\begin{figure}\centering
\includegraphics{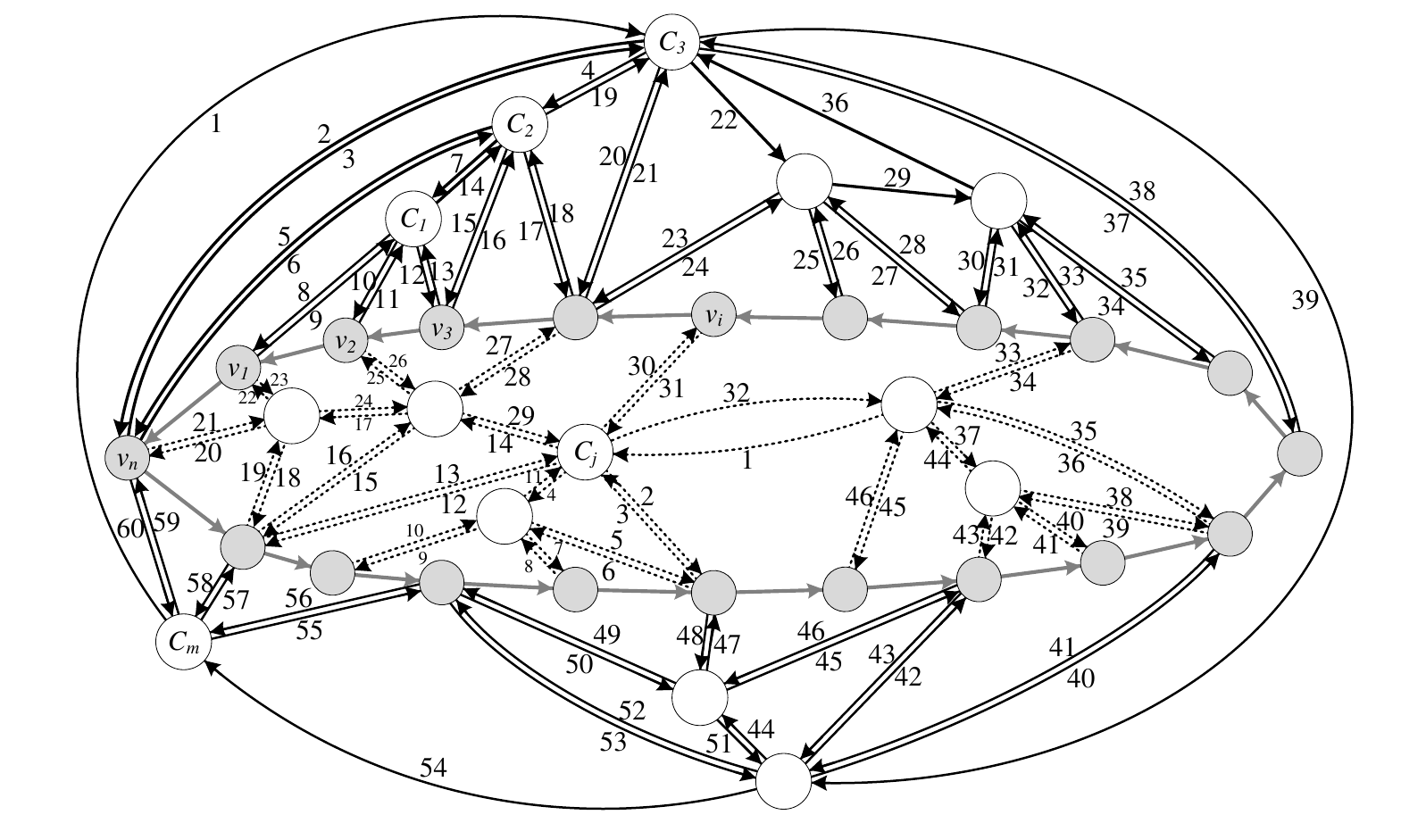}
\caption{The closed walks \wout, \win; the numbers indicate the order in which edges are traversed by the walks.}\label{3}
\end{figure}
\begin{figure}\centering
\includegraphics{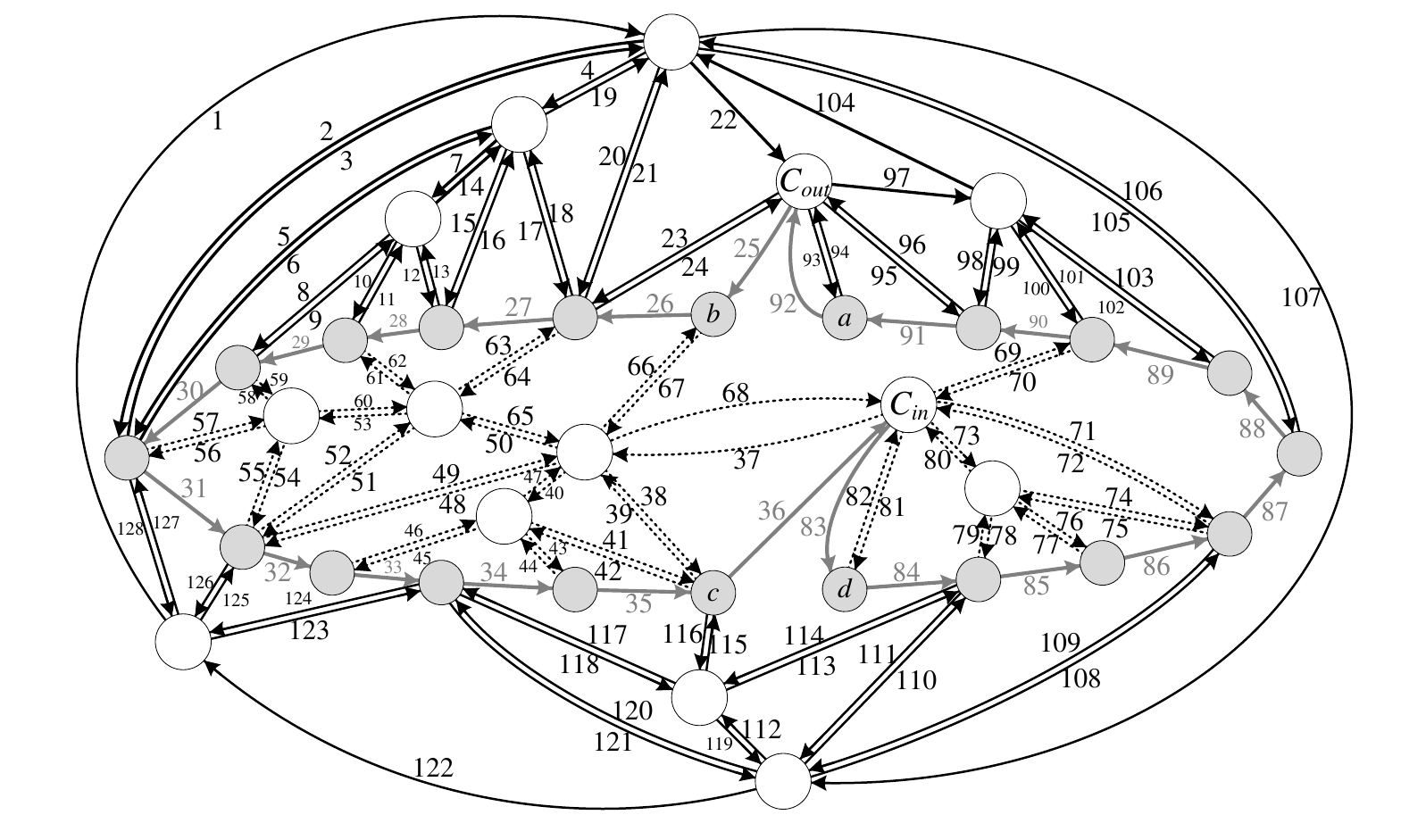}
\caption{The closed walk \w; the numbers indicate the order in which edges are traversed by the walk.}\label{4}
\end{figure}

\subsection*{Augmenting $I$}
The cycle $\mathbb{C}$ splits the plane into two parts; each clause belongs to exactly one part. We say that the clauses, edges, etc.\ inside (resp.\ outside) $\mathbb{C}$ are \e{inner} (resp.\ \e{outer}).

Focus on the outer clauses. Define parent-child relationship between the clauses as follows. The clauses that belong to the outer face of the graph $I$ are \e{orphans} -- they have no parents. Now imagine removing an orphan $C$, together with the edges that connect $C$ to the variables. Any clause $C'$ that now (after $C$'s removal) belongs to the outer face of $I$ is a \e{child} of $C$ (and $C$ is the \e{parent} of $C'$). Recursively, any clause $C''$ is the parent for all clauses that appear on the outer face of $I$ after removal of $C''$ (and all its ancestors).

We now augment $I$ with new edges. For any parent, the children are angularly sorted around the parent. We connect the first and the last child to the parent; if a parent has only one child, connect the parent and the child by 2 parallel edges (parallel in the graph-theoretic sense, in the embedding they are not parallel). Add also edges between siblings; in particular, connect orphans with a cycle, in the order in which they appear on the outer face of $I$ (Fig.~\ref{2}). Finally, add a parallel edge for each clause-variable edge (so that every clause is connected to each of its variables with 2 parallel edges).
\subsection*{The walks \wout,\win}
Let \wout be the closed walk that goes through the outer clauses and all variables in the DFS manner, with preference to go right (as seen from a clause) and down (i.e., towards $\mathbb{C}$). Specifically, at the "top level", the walk contains the cycle through all orphans. In addition, at each orphan $C$ (and in general, at each clause) the walk is the DFS traversal of the subtree of $C$: it goes to the rightmost (as seen from $C$) variable of $C$, then goes back to $C$ (using the parallel edge), and then -- either to the next variable of $C$ or to the rightmost child $C'$ of $C$ (whichever is more to the right). At $C'$, the walk recurses down to the rightmost variable of $C'$, then goes back to $C'$, and then again -- either to the next variable of $C'$ or to the rightmost child of $C'$, etc. The walk follows the edge to the sibling or to the parent of a clause $C''$ after all children and variables of $C''$ have been visited. In particular, the variables of childless clauses are just visited one-by-one, from right to left. Refer to Fig.~\ref{3}.

We do an analogous augmentation of the subgraph of $I$ inside the cycle $\mathbb{C}$. Specifically, we choose some face $F$ of $I$ inside $\mathbb{C}$, and let $F$ play the role of the outer face: The orphan clauses are those on the boundary of $F$; the children are defined recursively as the clauses that appear on $F$ after deletion of parents. As before, we duplicate variable-clause edges. Let \win be closed walk analogous to \wout: \win goes through the orphans, recursing to variables and children in the DFS manner, with the preference to go left (as seen from a clause) and towards $\mathbb{C}$. Refer to Fig.~\ref{3}.
\subsection*{The walk \w}
We now splice the walks \wout, \win, and the cycle $\mathbb{C}$ into a single closed walk \w through $I$. Let $ab$ be an arbitrary edge of $\mathbb{C}$, and let $C_{out}$ be an outer clause that belongs to the face of $I$ that has $ab$ on the boundary. Remove $ab$ from $\mathbb{C}$, and add edges $aC_{out},bC_{out}$ (Fig.~\ref{4}). Similarly, remove an edge $cd$ from $\mathbb{C}$, and add edges $cC_{in},dC_{in}$ to an inner clause $C_{in}$.

The walk \w starts from following \wout until reaching $C_{out}$, at which point it uses the edge $C_{out}b$. (This may not necessarily be the first time the walk visits $C_{out}$; the walk uses $C_{out}b$ so that the usage is consistent with the ordering of edges around $C_{out}$ -- refer to Fig.~\ref{4}, where $C_{out}$ is reached first by edge 22, while $C_{out}b$ is only 25th.) From $b$, the walk follows the cycle $\mathbb{C}$ up to $c$, where it uses the edge $cC_{in}$ to enter inside $\mathbb{C}$. From $C_{in}$, the walk \w follows the walk \win all the way around back to $C_{in}$. It then uses $C_{in}d$ to get back to $\mathbb{C}$, upon which it traverses $\mathbb{C}$ from $d$ to $a$. At $a$, the walk uses $aC_{in}$ and follows the rest of \wout from there.
\subsection*{From 3SAT to wire routing}
As the last step of the reduction, we convert $I$ to an instance of finding a thick wire. For that, we thicken the edges of $I$, turning them into channels of width~2. We replace variables and clauses with gadgets shown in Figs.~\ref{var} and~\ref{clause} resp. The connections (channels) between variables and clauses (Fig.~\ref{connections}) ensure that whenever a channel from a clause to a variable is used by the wire, the variable satisfies the clause.

We cut (the channel corresponding to) one of the edges of $\mathbb{C}$, and place the points $s,f$ on the opposite sides of the cut. This turns the closed walk \w into an $s\-f$ path. By our construction, the only $s\-f$ wire in the instance is one that follows the walk \w in $I$, possibly, omitting some clause-variable edges (channels). Indeed, neither channels nor gadgets have any leakage -- the wire must follow them through. The only flexibility that the wire has is (1)~how to traverse the variable gadgets, and (2)~which clause-variable channels (not) to use. But we know, by clause gadget construction (see Fig.~\ref{clause}), that at least one channel from every clause must be used; moreover, when it is used, the variable must satisfy the clause (see Fig.~\ref{connections}). Thus, the $s\-f$ wire exists if and only if $I$ is satisfiable.
\begin{figure}\centering
\includegraphics{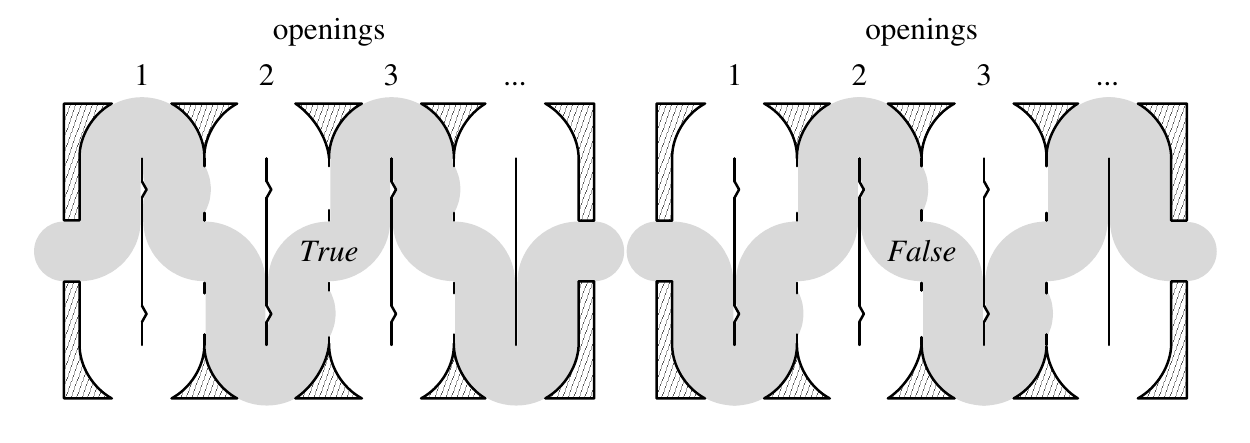}
\caption{A variable gadget can be traversed in one of the two ways, setting the truth assignment. Depending on the way, the path bulges out of the odd (when the variable is set to True) or even (when it is False) openings at the top of the gadget, and resp.\ even/odd openings at the bottom.}\label{var}
\end{figure}
\begin{figure}\centering
\includegraphics{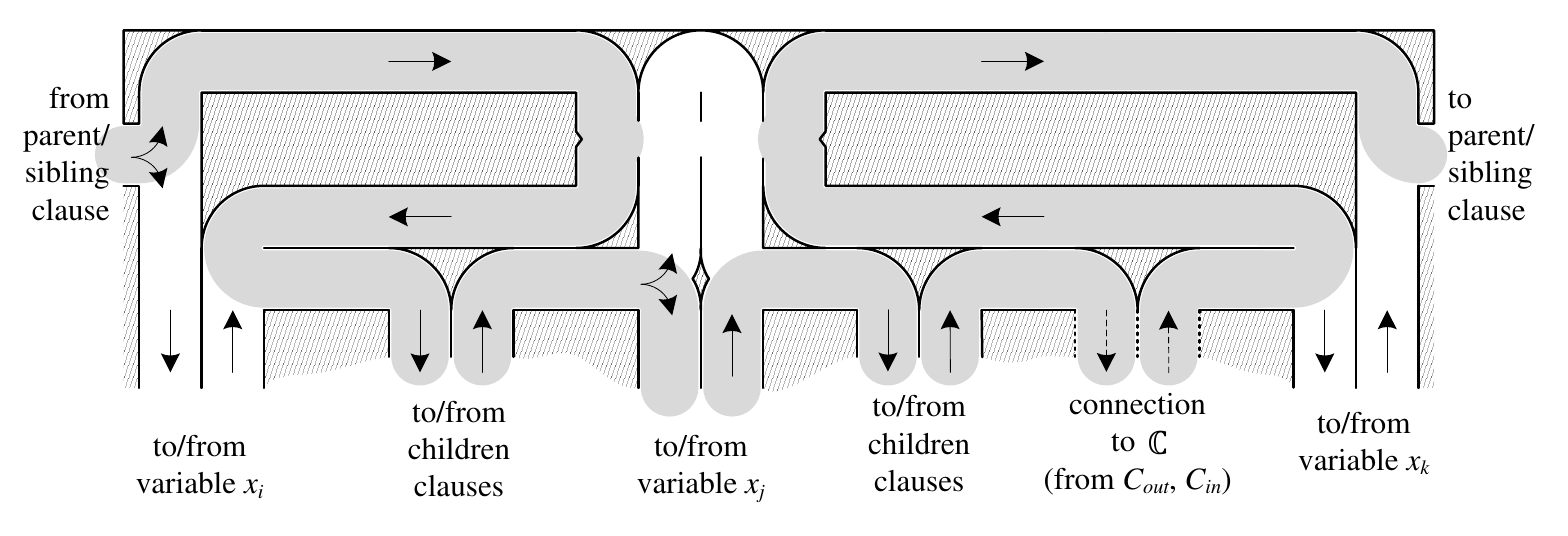}
\caption{When a clause gadget is traversed, from left to right, one of the channels leading to variables must be used. Otherwise, 3 subpaths go through the top of the gadget leading to a self-overlap.}\label{clause}
\end{figure}
\begin{figure}\centering
\includegraphics{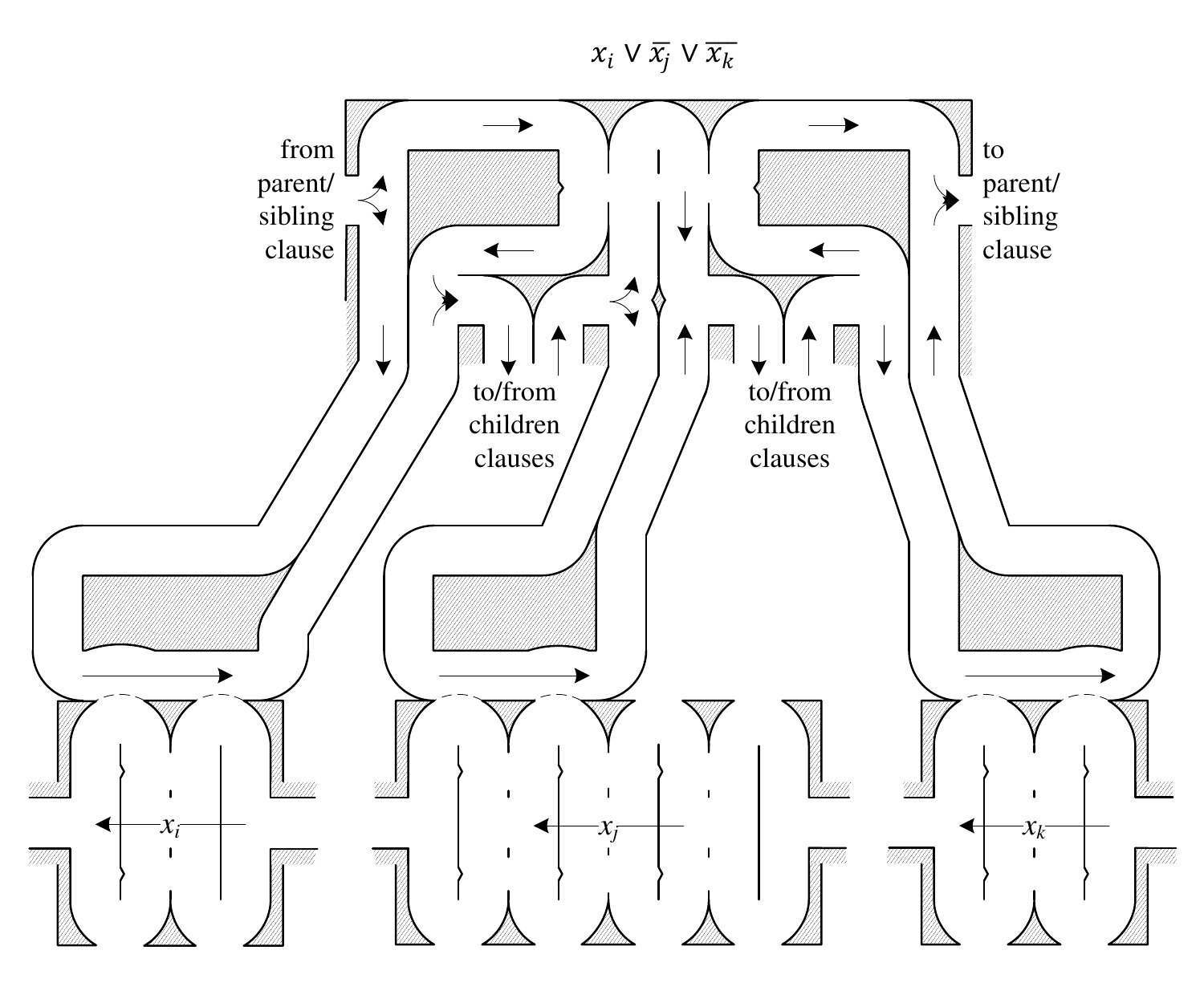}
\caption{If a variable does not satisfy a clause the channel between them cannot be used by non-selfoverlapping path.}\label{connections}
\end{figure}

\section{Conclusion}
We showed that for a snake to stay simple, it must grow fat (or else, learn to squeeze).
We mention few open problems here:
\begin{description}
\item To minimize snake's squeeze factor, it would be interesting to find a path that minimizes maximum self-overlap of a snake. It is easy to see that a snake of thickness 1/2 can follow (without self-overlap) the shortest (possibly self-overlapping) path for a thickness-1 snake. Is 1/2 best possible? Note that in our hardness proof, a non-zero amount of overlap is enforced; what is the largest self-overlap for which the problem remains hard?
\item Minimizing the path length for \e{one} point of the snake is similar to finding ``$d_1$-optimal'' motion of a rod \cite{rod}. Other objectives are possible: e.g., distance traveled by another point (not the mouth) or the average distance traveled ($d_2\-,\dots,d_\infty$-optimality, etc.). Does the situation with the snake mimic that for the rod: minimizing motion of any point is NP-hard, except for a rod endpoint?
\item What is the hardness of computing paths for a short snake that is not required to be pulled taut? It seems that the answer to this question leads to an interesting research direction of ``snake packing'': Given a polygonal domain, can one layout a length-\l snake in it? The problem is NP-hard by a reduction from Hamiltonicity of grid graphs; what about, say, simple polygons?
\item Can a shortest \e{rectilinear} wire be computed efficiently?
\end{description}
\paragraph{Acknowledgements} We thank Estie Arkin, David Kirkpatrick, Joe Mitchell and Jukka Suomela for discussions. This work was partially supported by the Academy of Finland grant 118653 (ALGODAN).


\begin{thebibliography}{10}

\bibitem{hippoAlon2}
P.~K. Agarwal, A.~Efrat, and M.~Sharir.
\newblock Vertical decomposition of shallow levels in 3-dimensional
  arrangements and its applications.
\newblock {\em SIAM Journal on Computing}, 29(3):912--953, 2000.

\bibitem{hippoPankaj}
P.~K. Agarwal and M.~Sharir.
\newblock Efficient algorithms for geometric optimization.
\newblock {\em ACM Computing Surveys}, 30:412--458, 1998.

\bibitem{needleWAFR}
R.~Alterovitz, M.~S. Branicky, and K.~Y. Goldberg.
\newblock Motion planning under uncertainty for image-guided medical needle
  steering.
\newblock {\em International Journal of Robotic Research},
  27(11-12):1361--1374, 2008.

\bibitem{arkin10maximum}
E.~M. Arkin, J.~S.~B. Mitchell, and V.~Polishchuk.
\newblock Maximum thick paths in static and dynamic environments.
\newblock {\em Computational Geometry Theory and Applications}, 43(3):279--294,
  2010.

\bibitem{rod}
T.~Asano, D.~Kirkpatrick, and C.~K. Yap.
\newblock $d_1$-optimal motion for a rod.
\newblock In {\em Proceedings of the 12th Annual ACM Symposium on Computational
  Geometry}, pages 252--263, 1996.

\bibitem{rodCCCG}
T.~Asano, D.~Kirkpatrick, and C.~K. Yap.
\newblock Minimizing the trace length of a rod endpoint in the presence of
  polygonal obstacles is {NP}-hard.
\newblock In {\em Proceeding of Canadian Conference on Computational Geometry},
  pages 10--13, 2003.

\bibitem{hippoanchored}
J.~Barcia, J.~Diaz-Banez, F.~Gomez, and I.~Ventura.
\newblock The anchored {Voronoi} diagram: static and dynamic versions and
  applications.
\newblock In {\em 19th European Workshop on Computational Geometry}, 2003.

\bibitem{corridors}
S.~Bereg and D.~Kirkpatrick.
\newblock Curvature-bounded traversals of narrow corridors.
\newblock In {\em Proceedings of the 21st Annual Symposium on Computational
  geometry}, pages 278--287, 2005.

\bibitem{chew}
L.~P. Chew.
\newblock Planning the shortest path for a disc in {$O(n^{2} \log n)$} time.
\newblock In {\em Proceedings of the 1st Annual Symposium on Computational
  Geometry}, pages 214--220, 1985.

\bibitem{protein}
D.~Chowdhury.
\newblock Molecular motors: Design, mechanism, and control.
\newblock {\em Computing in Science and Engineering}, 10:70--77, 2008.

\bibitem{needleCook}
A.~F. {Cook IV}, C.~Wenk, O.~Daescu, S.~Bitner, Y.~K. Cheung, and A.~Kurdia.
\newblock Visiting a sequence of points with a bevel-tip needle.
\newblock In {\em Proceedings of the 9th Latin American Theoretical Informatics
  Symposium}, 2010.

\bibitem{hippoAlon}
A.~Efrat and M.~Sharir.
\newblock A near-linear algorithm for the planar segment center problem.
\newblock {\em Discrete \& Computational Geometry}, 16:239--257, 1996.

\bibitem{hu}
T.~Hu, A.~Kahng, and G.~Robins.
\newblock Optimal robust path planning in general environments.
\newblock {\em IEEE Transactions on Robotics and Automation}, 9:775--784, 1993.

\bibitem{fun}
I.~Kostitsyna and V.~Polishchuk.
\newblock Simple wriggling is hard unless you are a fat hippo.
\newblock In {\em Fifth International Conference on Fun with Algorithms (FUN)},
  2010.
\newblock To appear.

\bibitem{latombe}
J.-C. Latombe.
\newblock {\em Robot Motion Planning}.
\newblock Kluwer, Boston, 1991.

\bibitem{lavalle}
S.~M. LaValle.
\newblock {\em Planning Algorithms}.
\newblock Cambridge University Press, 2006.

\bibitem{rod3d}
J.~Y. Lee and H.~Choset.
\newblock Sensor-based planning for a rod-shaped robot in three dimensions:
  Piecewise retracts of {R}3 x {S}2.
\newblock {\em International Journal of Robotic Research}, 24(5):343--383,
  2005.

\bibitem{p3sat}
D.~Lichtenstein.
\newblock Planar formulae and their uses.
\newblock {\em SIAM Journal on Computing}, 11(2):329--343, 1982.

\bibitem{maley}
F.~M. Maley.
\newblock {\em Single-Layer Wire Routing and Compaction}.
\newblock MIT Press, 1990.

\bibitem{vlsiBook}
K.~McEvoy and J.~V. Tucker, editors.
\newblock {\em Theoretical Foundations of VLSI Design}.
\newblock Cambridge University Press, New York, NY, USA, 1991.

\bibitem{hippoPach}
J.~Pach and G.~Tardos.
\newblock Forbidden patterns and unit distances.
\newblock In {\em Proceedings of the 21st Annual Symposium on Computational
  Geometry}, pages 1--9, 2005.

\bibitem{rodSharir}
S.~Sifrony and M.~Sharir.
\newblock A new efficient motion-planning algorithm for a rod in
  two-dimensional polygonal space.
\newblock {\em Algorithmica}, 2:367--402, 1987.

\bibitem{nature}
C.~Veigel, L.~M. Coluccio, J.~D. Jontes, J.~C. Sparrow, R.~A. Milligan, and
  J.~Molloy.
\newblock The motor protein myosin-{I} produces its working stroke in two
  steps.
\newblock {\em Nature}, 398:530--–533, 1999.

\end{thebibliography}
\end{document}